\documentclass[a4paper,UKenglish,cleveref, autoref]{lipics-v2021}
\usepackage{amssymb,amsmath,latexsym,amsthm}
\usepackage{xspace}
\usepackage{graphicx}
\usepackage{color} 
\usepackage{dsfont}
\usepackage{mathtools,environ}
\usepackage{thmtools,thm-restate,tikz}
\usetikzlibrary{arrows,automata,decorations.pathmorphing,positioning,shapes,shapes.geometric,fit,calc}

\makeatletter
\def\blfootnote{\xdef\@thefnmark{}\@footnotetext}
\makeatother

\newcommand{\ST}{~|~}
\newcommand{\Nat}{\mathds{N}}
\newcommand{\Rat}{\mathds{Q}}
\newcommand{\Reals}{\mathds{R}}

\newcommand{\tuple}[1]{\langle #1  \rangle}

\DeclarePairedDelimiter{\floor}{\lfloor}{\rfloor}
\DeclarePairedDelimiter{\ceil}{\lceil}{\rceil}

\newcommand{\TDS}{\textsf{TDS}\xspace}
\newcommand{\TDSzo}{\textsf{TDS}$_{01}$\xspace}
\newcommand{\CTDS}{\textsf{CTDS}\xspace}
\newcommand{\GTDS}{\textsf{GTDS}\xspace}
\newcommand{\CGTDS}{\textsf{CGTDS}\xspace}

\newcommand{\TDSF}{\textsf{TDS$^F$}\xspace}
\newcommand{\CTDSF}{\textsf{CTDS$^F$}\xspace}
\newcommand{\GTDSF}{\textsf{GTDS$^F$}\xspace}
\newcommand{\CGTDSF}{\textsf{CGTDS$^F$}\xspace}

\newcommand{\PAM}{\textsf{PAM}\xspace}

\newcommand{\DSA}{\textsf{DSA}\xspace}
\newcommand{\DSAs}{\textsf{DSA}s\xspace}

\newcommand{\A}{{\cal A}}
\newcommand{\B}{{\cal B}}
\newcommand{\C}{{\cal C}}
\renewcommand{\P}{{\cal P}}

\newcommand{\Dz}{{\ensuremath{\mathsf 0}}\xspace}
\newcommand{\Do}{{\ensuremath{\mathsf 1}}\xspace}
\newcommand{\Dt}{{\ensuremath{\mathsf 2}}\xspace}

\newcommand{\Dm}[1]{\ensuremath{a_{#1}}}

\newcommand{\Paragraph}[1]{\subsection*{#1}}
\newtheorem*{exampleE}{Example}

\newcommand{\Skip}[1]{}

\hideLIPIcs
\nolinenumbers 

\author{Udi Boker}{Reichman University, \\
Herzliya, Israel}
{udiboker@runi.ac.il}{0000-0003-4322-8892}{}

\author{Thomas A. Henzinger}
{IST Austria,\\
Klosterneuburg, Austria}
{tah@ist.ac.at}
{0000-0002-2985-7724}
{}

\author{Jan Otop}
{University of Wroc\l{}aw,\\
	Wroc\l{}aw, Poland}
{jotop@cs.uni.wroc.pl}
{0000-0002-8804-8011}
{}

\authorrunning{U. Boker, T.\,A. Henzinger, and J. Otop}

\ccsdesc[500]{Theory of computation~Formal languages and automata theory}

\keywords{Algorithms, Automata, Discounted-sum automata, Discrete mathematics}

\title{The Target Discounted-Sum Problem}

\begin{document}

\maketitle

\begin{abstract}
The \emph{target discounted-sum} problem is the following: Given a rational discount factor $0<\lambda<1$ and three rational values $a,b$, and $t$, does there exist a finite or an infinite sequence $w \in \{a,b\}^*$ or $w \in \{a,b\}^\omega$, such that $\sum_{i=0}^{|w|} w(i)  \lambda^i$ equals $t$?

The problem turns out to relate to many fields of mathematics and computer science, and its decidability question is surprisingly hard to solve.

We solve the finite version of the problem, and show the hardness of the infinite version, linking it to various areas and open problems in mathematics and computer science: $\beta$-expansions, discounted-sum automata, piecewise affine maps, and generalizations of the Cantor set.
We provide some partial results to the infinite version, among which are solutions to its restriction to eventually-periodic sequences and to the cases that $\lambda \geq \frac{1}{2}$ or $\lambda=\frac{1}{n}$, for every $n\in\Nat$.

We use our results for solving some open problems on discounted-sum automata, among which are the exact-value, universality and inclusion problems for functional automata.
\blfootnote{* This paper refines and expands the LICS 2015 paper with the same title (DOI 10.1109/LICS.2015.74), and in particular fixes the statements of Theorems 23, 25, and 26.}
\end{abstract}

\section{Introduction}
Discounting the influence of future events takes place in many natural processes, such as temperature change, capacitor charge, and effective interest rate, for which reason it is a key paradigm in economics and it is widely studied in game theory, Markov decision processes, and automata theory ~\cite{Andersson06,BH11,BH14,CFW13,AFHMS05,DiscountingInSystems,functional,FiniteValued,DiscountedMarkov}. Yet, 
the decidability questions of basic problems with regard to these models are still open. For example, the universality and inclusion problems of discounted-sum automata (see Section~\ref{sec:DSA}).

It turns out that the following simple \emph{target discounted-sum problem} constitutes the main difficulty in many of these open problems.
It was first raised, to the best of our knowledge, by Jean-Fran\c{c}ois Raskin in the course of analyzing discounted-sum automata. 

\begin{definition}[\TDS, \TDSF]\label{def:TDS}
Given a rational \emph{discount factor} $0<\lambda<1$, a \emph{target} rational value $t$, and rational \emph{weights} $a$ and $b$, the \emph{target discounted-sum problem} is the question of whether there exists a finite, resp.\ an infinite, sequence (\emph{solution}) $w \in \{a,b\}^*$, resp.\ $w \in \{a,b\}^\omega$, such that $\sum_{i=0}^{|w|} w(i)  \lambda^i$ equals $t$.
For distinguishing between the question about a finite and an infinite sequence, we denote the former problem \TDSF and the latter \TDS.
\end{definition}

Despite its simple statement, resolving the decidability of \TDS appears to be challenging and relates to many open questions in mathematics and computer science.

This problem is a natural milestone for open problems that involve discounting, such as problems on discounted-sum automata \cite{BH14,CDH08,functional,FiniteValued}, discounted-sum two-player games \cite{BMR14,HR2014}, and multi-objective discounted-sum reachability \cite{CFW13}. In particular, \TDS reduces to the universality problem of discounted-sum automata over infinite words (Theorem~\ref{thm:TdsDsa}), whereas the exact-value problem of functional discounted-sum automata reduces to a generalized version of \TDSF (Theorem~\ref{thm:ExactValue}). Using our solution to \TDSF, we solve some of these open problems in Section~\ref{sec:DSA}. We further describe them at the end of this section.

Less intuitive is the connection between \TDS and piecewise affine maps (Section~\ref{sec:PAM}). Its reachability problem asks, given a piecewise affine map $f$, and points $s$ and $t$ in a space of some dimension $d$, whether there exists $n\in\Nat$, such that $f^n(s)=t$. The problem is known to be undecidable for 2 or more dimensions \cite{KCG94,Moo90,Moo91}, and it is open for one dimension \cite{AMP95,AG02,KPC08}. (The reachability problem of an affine map, which is not ``piecewise'', is known as the ``orbit problem'', and it is decidable for all dimensions over the rationals \cite{KL86,Sha79}. The orbit problem also relates to some decidable questions on linear recurrence sequences \cite{EPST03,OW14}.) We show that \TDS reduces to the reachability problem of one-dimensional piecewise affine maps (Theorem~\ref{thm:TdsPam}).

Another interesting connection is between \TDS and generalized Cantor sets (Section~\ref{sec:CS}). The standard Cantor set is obtained by starting with the interval $[0,1]$, and removing, at each iteration, the middle third of the remaining intervals. 
A natural generalization is to remove the middle $k$th (for example, the middle fifth) of the remaining intervals at each iteration \cite{Din01,Eid05,GR95}. While the membership question of a given number in the standard Cantor set is easily resolved, this is not the case with the general ``middle $k$th Cantor set''. The reason for the difference between removing the middle third and removing, for example, the middle fifth, lies in the fractal behavior of these removals: In the former case, each of the remaining intervals constitutes a third of the original interval, and should thus be multiplied by $3$ so as to view it as the original problem. In the latter case, each of the remaining intervals constitutes $\frac{2}{5}$ of the original interval, and should thus be multiplied by $\frac{5}{2}$ in order to view it as the original problem. This is closely related to number representation in an integral base, such as $3$, which is very simple, as opposed to representation in a nonintegral base, such as $\frac{5}{2}$, which is significantly more complicated (Section~\ref{sec:TdsAsBetaExp}). We show that the membership problem in the middle $k$th Cantor set reduces to \TDS (Theorem~\ref{thm:TdsCs}).

Analyzing \TDS, we show that it can be reduced to its restricted version, denoted by \TDSzo, in which the weights are fixed to be $0$ and $1$. The importance of this reduction is that it allows to view \TDS as a question within the well studied area of $\beta$-expansions ($\beta$-representations), which deals with the representation of numbers in a nonintegral base \cite{Ren57}. Namely, a \TDSzo instance with a discount factor $\lambda$ and a target value $t$ has a solution iff $t$ has a representation in base $\frac{1}{\lambda}$, using only the \Dz and \Do digits.

Unfortunately, though, the structure of $\beta$-expansions is still largely a mystery. Nikita Sidorov wrote \cite{Sid13}, for example: ``Usually the (greedy) expansions in bases like $\frac{3}{2}$ are considered virtually impossible to work with. For instance, if you expand $\frac{2}{5}$, say, in base $\frac{3}{2}$, then virtually nothing is known about this expansion.'' 

Nevertheless, analyzing \TDS through $\beta$-expansions leads to some partial results to the problem. An immediate corollary is the decidability for discount factors equal to or bigger than half. This is because every number has a representation in base $1<\beta \leq 2$, using only the \Dz and \Do digits \cite{Ren57}. Other straightforward results are the uniqueness of the solution for a discount factor smaller than half, when a solution exists, and the co-recursively-enumerability of \TDS.

A key tool in our analysis of $\beta$-representations is the notion of ``gaps'': One can explore the representation of a given number in base $\beta$ using the ``greedy exploration'' \cite{Ren57} -- going from left to right, and adding at each position the maximal possible digit. In this process, every step produces a ``remainder'', which should be treated in the next position. A ``gap'' is intuitively a ``normalized remainder''; It is equal, at the $n$th step of the exploration, to the multiplication of the remainder and $\beta^n$. By normalizing the remainders into gaps, the exploration process has a fractal nature, which gives the intuition to its close connection to general Cantor sets and to piecewise affine maps. It is also related to generalizations of Mahler's $\frac{3}{2}$ problem \cite{Mah68} and Collatz's problem \cite{Con72,Leh08}, though more loosely.

Analyzing the possible gaps in the exploration process allows us to solve, in PSPACE, the restriction of \TDS to eventually-periodic sequences. As a special case, we get decidability for every discount factor of the form $\frac{1}{n}$, where $n\in\Nat$. One might be tempted to conclude that another special case is a solution to finite sequences, namely to \TDSF. This is, however, not the case, as an instance of \TDSF cannot be reduced, in general, to a \TDS instance with only \Dz's and \Do's.

\Skip{As for non-eventually-periodic representations, we do not even know whether, for the specific base $\frac{5}{2}$, there exists a rational number with a non-eventually-periodic representation that only uses \Dz's and \Do's. We thus look into a somewhat opposite question -- given a non-eventually-periodic representation, does it represent a rational number? We show that in some cases, for example when the growth rate of blocks of identical digits is more than exponential, the answer is negative.
} % end of \skip

\Skip{Looking into restricted versions of \TDS, by fixing some of its parameters, can only ease the problem when fixing the discount factor -- the cases that $\lambda\geq\frac{1}{2}$ and $\lambda=\frac{1}{n}$, where $n\in\Nat$, are shown to be decidable. Fixing the weights or fixing the target value leaves the problem as hard as the original one.}

For leveraging our result on eventually-periodic sequences into a solution to \TDSF and into new results about discounted-sum automata, we consider its following three natural extensions. The first is the generalization of \TDS to have arbitrarily many weights, denoted by \GTDS; the second is adding a parameter to the problem, constraining the sequence of weights by an $\omega$-regular expression, denoted by \CTDS; and the third is their combination, denoted by \CGTDS. The corresponding generalizations of \TDSF are denoted by \GTDSF, \CTDSF (getting a regular expression), and \CGTDSF.

Fortunately, we are able to generalize our results to all of these extended versions. 
The main challenge here is that, as opposed to \TDS, we no longer have the ``dichotomy'' property, stating that each discount factor either guarantees a solution or guarantees that the solution, if exists, is unique. The underlying reason is that when allowing enough digits, a number might not have, or may have many, and even infinitely many, different representations in a nonintegral base \cite{EJK94,GS01}. Nevertheless, we derive the generalized results using nondeterministic explorations, rather than greedy explorations, K\"onig's Lemma, and a few other observations.

Using our results, we solve some open problems on discounted-sum automata over finite words, among which are the exact-value, universality and inclusion problems for functional automata (Theorems~\ref{thm:ExactValue}, \ref{thm:FunctionalInclusion}.)

\section{Problem Restrictions and Extensions}\label{sec:Problem}
We analyze below the natural restrictions and extensions of \TDS (Definition~\ref{def:TDS}). Most interesting are the restriction to fixed weights of $0$ and $1$, which will be shown to be equivalent to the original problem, and the generalization to arbitrarily many weights, for which all the positive results will follow. 

\Paragraph{Restrictions}
One may wonder whether some restricted versions of \TDS are easier to solve. A natural approach in this direction is to fix one of the four \TDS parameters.

Fixing the discount factor may indeed ease the problem. We will show solutions to the cases where $\lambda \geq \frac{1}{2}$ (Theorem~\ref{thm:BiggerThanHalf}) and  $\lambda = \frac{1}{n}$, for every natural number $n$ (Theorem~\ref{thm:IntegralDenominator}).

Fixing the target value, on the other hand, cannot help: The general problem, with a discount factor $\lambda$, a target $t$, and weights $a$ and $b$, can be reduced to a restricted problem that only allows a constant target value $T$, by choosing new weights $a'=a\cdot\frac{T}{t}$ and $b'=b\cdot\frac{T}{t}$.

Fixing the weights also cannot ease the problem. We show below that even when the weights are restricted to be exactly $0$ and $1$, the problem remains exactly as hard as in the general case. This observation is the key for approaching \TDS via $\beta$-expansions (Section~\ref{sec:TdsAsBetaExp}). To fit into the $\beta$-expansion setting, we formally define the $0$-$1$ restriction of the problem to start the sequence of summations $\sum_{i=0}^\infty w(i)  \lambda^i$ with $\lambda^1$ rather than with $\lambda^0$, i.e., we put $w(0) = 0$.

\begin{definition}[\TDSzo]\label{def:TDSzo}
Given a rational \emph{discount factor} $0<\lambda<1$ and a \emph{target} rational value $t$, the \emph{$0$-$1$ target discounted-sum problem} (\TDSzo) is the question of whether there exists an infinite sequence (\emph{solution}) $w \in \{0,1\}^\omega$, such that $\sum_{i=1}^\infty w(i)  \lambda^i$ is equal to $t$.
\end{definition}

We show below that \TDSzo is exactly as hard as the general infinite \TDS.

\begin{restatable}{theorem}{TdsToTdszo}
\label{thm:TdsToTdszo}
\TDS reduces to \TDSzo. 
The reduction preserves eventual periodicity and non-eventual periodicity of the solution. 
\end{restatable}
\begin{proof}
We make the reduction in three steps: 
\begin{enumerate}
\item We claim that every \TDS $\P$ with rational discount factor $0<\lambda<1$, target value $t$, and  weights $a$ and $b$,
is equivalent to the \TDS $\P'$ with discount factor $\lambda$, target $t'=t - \frac{a}{1-\lambda}$, and weights $a'=0$ and $b'=b-a$. Indeed, subtracting $a$ from every element in a discounted-sum sequence $w\in\{a,b\}^\omega$, provides a discounted-sum sequence $w'\in\{0,b-a\}^\omega$, such that 
$\sum_{i=0}^\infty w'(i)  \lambda^i = \sum_{i=0}^\infty w(i)\lambda^i - \sum_{i=0}^\infty a  \lambda^i = \sum_{i=0}^\infty w(i) \lambda^i - \frac{a}{1-\lambda}$.
%$w'=w- a  \sum_{i=0}^\infty \lambda^i = w - \frac{a}{1-\lambda}$. 
\item We claim that the \TDS $\P'$ is equivalent to the \TDS $\P''$ with discount factor $\lambda$, target $t''= \frac{t'}{b'} = \frac{t-\lambda t -a}{(1-\lambda)(b-a)}$, and weights $a''=0$ and $b''=1$. Indeed, dividing every element in a discounted-sum sequence $w'\in\{0,b'\}^\omega$ by $b'$, provides a discounted-sum sequence $w''\in\{0,1\}^\omega$, such that 
$\sum_{i=0}^\infty w''(i)  \lambda^i = \frac{1}{b'} \sum_{i=0}^\infty w'(i)  \lambda^i$.
%$w''=w'/b'$. 
\item The \TDS $\P''$ already uses the weights $0$ and $1$. It is obviously equivalent to the \TDSzo with a target $t''' = \lambda t'' = \frac{\lambda(t-\lambda t -a)}{(1-\lambda)(b-a)}$, starting the summation with $\lambda^1$ rather than with $\lambda^0$.
\end{enumerate}
\end{proof}

Observe that \TDSzo is more general than the corresponding problem with respect to finite sequences, as
every finite sequence $w$ can be considered as the infinite sequence $w \Dz^{\omega}$. Yet, \TDSF is not subsumed by \TDS, and cannot be reduced to only have the $0$ and $1$ weights. (The first step in the proof of Theorem~\ref{thm:TdsToTdszo} only holds for infinite sequences.) We shall return to \TDSF at the end of the section.

\Paragraph{Extensions}
We consider two natural extensions of \TDS and \TDSF, as well as their combination. 

The first generalization allows for arbitrarily many weights:
\begin{definition}[\GTDS]\label{def:Gtds}
Given a rational \emph{discount factor} $0<\lambda<1$, a \emph{target} rational value $t$, and rational \emph{weights} $a_1, \ldots, a_k$, for $k\in\Nat$, the \emph{generalized target discounted-sum problem} (\GTDS) is the question of whether there exists an infinite sequence (\emph{solution}) $w \in \{a_1,\ldots, a_k\}^\omega$, such that $\sum_{i=0}^\infty w(i) \lambda^i$ is equal to $t$.
\end{definition}

The second extension adds an $\omega$-regular constraint on the allowed sequences. Such a constraint is particularly relevant for linking between \TDS and \TDSF, as well as in the scope of discounted-sum automata (Section~\ref{sec:DSA}). 

\begin{definition}[\CTDS]\label{def:TdsReg}
Given a rational \emph{discount factor} $0<\lambda<1$, a \emph{target} rational value $t$, rational \emph{weights} $a$ and $b$, and an $\omega$-regular expression $e$, the \emph{constrained target discounted-sum problem} (\CTDS) is the question of whether there exists  an infinite sequence (\emph{solution}) $w \in \{a,b\}^\omega$, such that $\sum_{i=0}^\infty w(i)  \lambda^i$ is equal to $t$ and $w$ belongs to the language of $e$.
\end{definition}

Note that \CTDS is a proper extension of \TDS. Indeed, a special variant of \CTDS with $e$ defined as $(a + b)^{\omega}$
is equivalent to \TDS.

One may then consider the combination of \GTDS and \CTDS, denoted by \CGTDS, allowing arbitrarily many weights, and imposing a regular constraint on the allowed sequences.

All of the above extensions also apply to \TDSF, denoted by \GTDSF, \CTDSF (getting a regular expression), and \CGTDSF, respectively.
By allowing arbitrarily many weights and a constraint, the finite version is subsumed by the infinite version of the problem, as formalized in the following theorem.

\begin{restatable}{theorem}{FiniteTdsToCgtds}
\CGTDSF reduces to \CGTDS.
\label{t:FiniteTdsToCgtds}
\end{restatable}
\begin{proof}
Consider an instance $\P$ of \CGTDSF  with a discount factor $\lambda$, weights $a_1, \ldots, a_k$, a target $t$, and a regular expression $e$. 
We define an instance $\P'$ of \CGTDS as follows: if one of $a_1, \ldots, a_k$ is $0$ then
$\P'$ is the same as $\P$, except for having an $\omega$-regular expression $e' := e \cdot 0^{\omega}$.
Otherwise, we define $\P'$ to be the same as $\P$, except for having an additional weight of value $0$, and the $\omega$-regular expression $e' := e \cdot 0^{\omega}$, where $e$ does not contain the additional weight. Then, $\P$ has a finite solution if and only if $\P'$ has an infinite solution. 
\end{proof}

\section{\TDS as a Question on $\beta$-Expansions}\label{sec:TdsAsBetaExp}
Once reducing \TDS to \TDSzo (Theorem~\ref{thm:TdsToTdszo}), we can address the problem via $\beta$-expansions, namely by representing numbers in a nonintegral base. 

\Paragraph{Nonintegral base} We are used to represent numbers in an integral base (radix) -- decimal, binary, hexadecimal, etc. For example, the string ``3.56'' in decimal base is equal to $3\cdot 10^0 + 5\cdot 10^{-1}+6\cdot 10^{-2}$. Yet, the representation may be with an arbitrary base $\beta > 1$, in which case the string ``3.56'' is equal to $3\cdot \beta^0 + 5\cdot \beta^{-1}+6\cdot \beta^{-2}$. Representation in a nonintegral base is known as \emph{$\beta$-expansion}, a notion introduced by R\'{e}nyi \cite{Ren57} and first studied in the seminal works of R\'{e}nyi and of Parry \cite{Par60}.

The representation may be finite, as above, or infinite, as in the cases of representing the number $\frac{1}{3}$ in decimal base by the string ``0.33333\ldots'', and representing the number $\frac{10}{21}$ in base $\frac{5}{2}$ by ``$0.10101010\ldots$''. 

We denote the value of a representation $w$ in base $\beta$ by $w_{[\beta]}$ (and when no base is mentioned, it is the decimal base). For example, the value of $0.102_{[\frac{5}{2}]}$ is $\frac{66}{125}$ (as it equals to $1 \cdot \frac{2}{5} + 2 \cdot \frac{8}{125}$).

When dealing with a nonintegral base $\beta >1$, the representations are still required to only contain digits that stand for natural numbers. A well known result \cite{Ren57} is that all real numbers have a $\beta$-representation, using the numbers $\{0,1,2,\ldots,\ceil{\beta-1} \}$ as digits. In general, there might be several, and even infinitely many, representations to the same number \cite{EJK94}. (For an analysis of the numbers with unique representations see \cite{GS01}.)

\Paragraph{Greedy and lazy explorations} The most common scheme for generating a representation for a given number is the \emph{greedy exploration} (\emph{greedy expansion}) \cite{Ren57}, going from left to right, and adding at each position the maximal possible digit. For example, when representing $\frac{5}{8}$ in binary, we start with the digit \Do for the value $\frac{1}{2}$, getting a \emph{remainder} of $\frac{1}{8}$. We cannot continue with another \Do, since $\frac{1}{4}>\frac{1}{8}$, so we put \Dz (and the remainder is still $\frac{1}{8}$). We then put \Do for $\frac{1}{8}$, and we are done, as the remainder is $0$, getting the representation $.101_{[2]}$. Note that the greedy exploration provides the largest possible representation, lexicographically. Another common scheme is the \emph{lazy exploration} (\emph{lazy expansion}), adding at each position the minimal possible digit, providing the minimal possible representation, lexicographically.
Finally, a unifying scheme for greedy and lazy explorations is the \emph{nondeterministic exploration}, nondeterministically choosing at each step an 
eligible digit, i.e., a digit such that the remainder is non-negative and is smaller than the maximal number that can be represented starting from the current position.

\Paragraph{\TDSzo -- reformulated} \TDSzo can be naturally written as a question about representing a number in a nonintegral base: 

\begin{proposition}
A \TDSzo instance with a discount factor $\lambda$ and a target value $t$ has a solution iff $t$ has a representation in base $\frac{1}{\lambda}$, using only the \Dz and \Do digits.
\end{proposition}

Motivated by \TDSzo, we only consider rational bases. In addition, as it is trivial to decide the \TDSzo problem for a discount factor $\lambda \geq \frac{1}{2}$ (Theorem~\ref{thm:BiggerThanHalf}), which relates to a base $\beta \leq 2$, we will mostly consider a base $\beta=\frac{p}{q}>2$, and assume that $p$ and $q$ are co-prime.

\section{\TDS\ -- Analysis and Results}\label{sec:TdsResults}
We handle in this section the target discounted-sum problem (\TDS) and its constrained version (\CTDS).

We start with an immediate corollary of viewing \TDS as a question of $\beta$-expansions, providing a solution to the case that the discount factor is bigger than half.
We continue with defining ``gaps'' -- an alternative notion to the remainders that are maintained in exploring the representation of a number. The gaps will be a key tool in our analysis of the representations.

Unlike the general case of representing a rational number in a nonintegral base, we show that once the representations only use the \Dz and \Do digits, the representation, if exists, is unique for every rational in every base $\beta > 2$. A direct corollary is that \TDS is co-recursively enumerable.
A more delicate analysis of the possible representations allows us to decide whether a given rational has an eventually-periodic representation. As a special case, we get that \TDS is decidable for every discount factor in the form of $\frac{1}{n}$, where $n\in\Nat$.

All of the above results, except for the case that $\lambda\geq\frac{1}{2}$, also hold for \CTDS.

We conclude with looking into an opposite question -- given a non-eventually-periodic representation, does it represent a rational number? We show that in some cases, for example when the growth rate of blocks of identical digits is more than exponential, the answer is negative.

\Paragraph{The case of $\lambda \geq \frac{1}{2}$}
An immediate benefit of reducing \TDS to a question on $\beta$-representations is the solution for $\lambda$'s greater than $\frac{1}{2}$.

\begin{restatable}{theorem}{BiggerThanHalf}
\label{thm:BiggerThanHalf}
\TDS is decidable for every discount factor $\lambda\geq\frac{1}{2}$.
\end{restatable}
\begin{proof}
By Theorem~\ref{thm:TdsToTdszo} and Section~\ref{sec:TdsAsBetaExp}, we view the problem as asking whether a target number $t$ has a representation in base $\beta=\frac{1}{\lambda}$ with only \Dz's and \Do's. 

Thus, there are two simple cases: 
\begin{itemize}
\item The target $t$ is bigger than $\frac{1}{\beta-1}$, implying that $t$ does not have a $\beta$-representation, since $0.1^\omega_{[\beta]} = \frac{1}{\beta-1}<t$.
\item The target $t$ is equal to or smaller than $\frac{1}{\beta-1}$, implying that $t$ has a $\beta$-representation: Every number has a $\beta$-representation with the digits $\{0,1,2,\ldots, \ceil{\beta-1} \}$ \cite{Ren57}, and since $\lambda\geq \frac{1}{2}$ and $\beta=\frac{1}{\lambda}$, it follows that $\beta\leq 2$, meaning that the \Dz and \Do digit suffice. 
\end{itemize}
\end{proof}

The simplicity of solving \TDS for $\lambda \geq \frac{1}{2}$ suggests to check whether it also holds when slightly extending \TDS with a regular constraint; i.e., \CTDS. Yet, a slight extension, such as requiring a \Dz in every odd position, breaks the solution, as it is analogous to \TDS over $\lambda^2$.

\begin{example}\label{exm:RegExp}
Consider a \TDSzo instance with $\lambda = \frac{2}{3}$. One can apply Theorem~\ref{thm:BiggerThanHalf} for solving it, as $\frac{2}{3} \geq \frac{1}{2}$.
However, Theorem~\ref{thm:BiggerThanHalf} cannot hold for a \CTDS with the same parameters and an additional constraint that the sequence of \Dz's and \Do's is in the language of $(0(0+1))^{\omega}$. Indeed, $\sum_{i=0}^{\infty} w(i) \lambda^i = \sum_{i=0}^{\infty} w(2i) \lambda^{2i}$, which is a \TDSzo problem for $\lambda = (\frac{2}{3})^2 = \frac{4}{9} < \frac{1}{2}$.
\end{example}

\Paragraph{Gaps} 
In the exploration schemes presented in Section~\ref{sec:TdsAsBetaExp}, we compared at each step the remainder with the value of a digit in the currently handled position. For example, in binary representation, the value of the digit \Do in the first position to the right of the radix point is $\frac{1}{2}$, while its value in the third position is $\frac{1}{2} \cdot (\frac{1}{2})^2=\frac{1}{8}$. Analogously, we can always compare a fixed value of the digit, say $\frac{1}{2}$ for the digit \Do in binary, with the ``normalized-remainder''. In the above example, instead of checking whether $\frac{1}{2} \cdot (\frac{1}{2})^2 \leq\frac{1}{8}$, we would check whether  $\frac{1}{2}\leq \frac{1}{8} \cdot (2)^2$. We call this ``normalized-remainder'' the \emph{gap}. (A similar notion is used in \cite{BH11}.) 

When exploring a representation and working with gaps, we do not need to multiply, at each step, the remainder with the position-exponent of the base (as demonstrated above when comparing $\frac{1}{2}$ and $\frac{1}{8} \cdot (2)^2$). Instead, we can compute the new gap based on the current gap and the chosen digit, completely ignoring the remainder and the current position.

More formally, a left-to-right exploration in base $\beta$ (with digits only to the right of the radix point), using gaps, is done as follows: The initial gap is the target number, and in every step with a gap $g$, the gap of the next step is $g'=\beta g - m$, where $m$ is the chosen digit. (In case of the greedy exploration, it is the biggest $m$, such that $\frac{m}{\beta} \leq g$.) 

There are two main advantages for exploring representations using gaps rather than remainders:
\begin{itemize}
\item There is no need to store the current position, which generally grows to infinity, but only the current gap. In some cases, for example when the base $\beta$ is in the form of $\frac{1}{n}$ where $n\in\Nat$, there are only boundedly many possible gaps (Theorem~\ref{thm:IntegralDenominator}).
\item Once having a gap $g$ in some position $p$, the suffix of the representation (the string from position $p$ onwards) only depends on the gap $g$, and not on the position $p$. This is a central tool in analyzing the representation of a number, for example allowing to infer that the representation is eventually periodic when getting the same gap twice (Lemma~\ref{lem:EventuallyPeriodicGaps}).
\end{itemize}

\Paragraph{Unique representation with \Dz's and \Do's} 
With a nonintegral base $\beta$, a number may have several, and even infinitely many, $\beta$-representations \cite{GS01,Ren57}. In the following, we show that if  $\beta>2$ and the representation can only use the \Dz and \Do digits, then the representation, if exists, is unique.

\begin{restatable}{lemmaStatement}{UniqueZO}
\label{lem:UniqueZO}
Consider a base $\beta>2$ and a number $t$. If $t$ has a $\beta$-representation with only \Dz's and \Do's then this representation is unique.
\end{restatable}
\begin{proof}
Assume towards contradiction that a number $t$ has two such $\beta$-representations $w$ and $w'$.
As $w \neq w'$, there is a finite word $u$ and infinite words $v$ and $v'$, such that $w = \Dz.u\Dz v$ and
$w' = \Dz. u \Do v'$. 
Now, observe that $w_{[\beta]} \leq \Dz.u \Dz \Do^{\omega}_{[\beta]} =  \Dz.u_{[\beta]} + (\frac{1}{\beta})^{|u|+2} (\frac{\beta}{\beta-1})$, which is strictly smaller than
$w'_{[\beta]} \geq \Dz.u \Do \Dz^{\omega}_{[\beta]} = \Dz.u_{[\beta]} + (\frac{1}{\beta})^{|u|+1}$, leading to a contradiction.
\end{proof}

\Paragraph{The possible scenarios of the exploration}
Establishing that the representation, if exists, is unique, there are exactly three possible scenarios when running the greedy exploration with a rational base $\beta$ and a target number $t$: 
\begin{itemize}
\item The exploration stops when reaching a gap $g$, such that $g > \frac{1}{\beta-1}$. The conclusion is that there is no representation of $t$ in base $\beta$, since even if we use, from this position further, only the \Do digit, the resulting number will be too small, as $\sum_{i=1}^\infty \beta^i= \frac{1}{\beta-1}< g$.
\item The exploration stops when reaching a gap $g$, such that $g$ already appeared as the gap in a previous step. The conclusion is that there is an eventually-periodic representation of $t$ in base $\beta$, as formalized in Lemma~\ref{lem:EventuallyPeriodicGaps}.
\item The exploration never stops, which happens in the case that $t$ has a $\beta$-representation that is not eventually periodic.
\end{itemize}

Using the observation on the possible scenarios, we get the following result.
\begin{restatable}{proposition}{TDSisCoRE}
\label{prop:TDSisCoRE}
\TDS is co-recursively-enumerable.
\end{restatable}

\Paragraph{Eventually-periodic representations} 
A key question is whether a number has an eventually-periodic representation. 
It is known that with every integral base (such as decimal, binary, etc.), a real number $n$ has an eventually-periodic representation if and only if $n$ is rational. This no longer holds when the base is not integral. While an eventually-periodic representation implies a rational number, there are also rationals with non-eventually-periodic representations.
Further, Schmidt \cite{Sch80} showed that if all rationals have eventually periodic greedy expansions in some base $\beta>1$, for a real number $\beta$, then $\beta$ is an algebraic integer, and more precisely, it is a Pisot number or a Salem number. (The opposite direction is an open problem, known as ``Schmidt's conjecture'' \cite{Har06}.) The only algebraic integers among the rational numbers are the integers, implying that if all rationals have eventually periodic greedy expansions in a rational base $\beta$ then $\beta$ is an integer.

We show below how to decide whether a given number has an eventually-periodic representation in a rational base $\beta>2$, when the representations can only use the \Dz and \Do digits.

We first formalize the direct connection between eventually-periodic representations and a bounded set of gaps in the exploration.
\begin{restatable}{lemmaStatement}{EventuallyPeriodicGaps}
\label{lem:EventuallyPeriodicGaps}
Consider a rational number $t$ and a rational base $\beta>2$. Then,
$t$ has an eventually-periodic $\beta$-representation with only \Dz's and \Do's  iff there are finitely many different gaps in the greedy-exploration of $t$.
\end{restatable}
\begin{proof}\

\begin{itemize}

\item 
Assume that there are finitely many different gaps in the greedy-exploration of $t$. Then, there is a step $i$ in the exploration in which we get a gap $g$, such that $g$ was also the gap in some step $j<i$. Hence, the exploration will use in step $i+1$ the same digit that was used in step $j+1$, and by induction, the used digits will be eventually periodic, providing an eventually-periodic representation.

\item

As for the other direction, assume that $t$ has an eventually-periodic representation, in which the repeated sequence of digits is the finite word $u$, starting in position $p$. By Lemma~\ref{lem:UniqueZO}, $t$ has a unique $\beta$-representation, implying that the greedy exploration should provide this representation. Now, let $g_p$ be the gap of the exploration in position $p$ and $g_{p+|u|}$ be the gap of the exploration in position $p+|u|$. By the definition of exploration gaps, we have $g_p=u^\omega_{[\beta]}=g_{p+|u|}$, leading to finitely many different gaps in the greedy-exploration of $t$.
\end{itemize}
\end{proof}

We show below that by analyzing the gaps of the greedy exploration, we can decide whether there exists an eventually-periodic representation.

\begin{restatable}{lemmaStatement}{EP}
\label{lem:EP}
For a rational number $t=\frac{a}{b}$ and a rational base $\beta=\frac{p}{q}>2$, where $a,b,p$ and $q$ are integers, we can decide in space polynomial in the binary representation of $a,b,p$ and $q$, whether $t$ has an eventually-periodic $\beta$-representation with only \Dz's and \Do's.
Moreover, the eventually-periodic $\beta$-representation of $t$, if it exists, is of the form $u v^{\omega}$ with $|u| + |v| \leq b$.
\end{restatable}
\begin{proof}
We consider the greedy exploration, and analyze the gap $g=\frac{c}{d}$ (where $c$ and $d$ are co-prime) at every step of the exploration. Note that having a gap $g=\frac{c}{d}$, and adding a digit $m$, the next gap is $g'= \frac{p}{q} \cdot \frac{c}{d} - m  = \frac{pc - qmd}{qd}$. Let $\beta=\frac{p}{q}$, where $p$ and $q$ are co-prime.

We prove the required decidability by showing the following two claims:
\begin{enumerate}
\item If the gap, $g=\frac{c}{d}$, is such that $c$ is not divisible by $q$ then the representation cannot be eventually periodic.
\item For an initial gap $g=\frac{c}{d}$, after at most $d$ exploration steps, the gap $\hat{g}=\frac{\hat{c}}{\hat{d}}$ will either 
\begin{itemize}
	\item exceed $\frac{1}{\beta-1}$ (implying that there is no representation, as $1^\omega_{[\beta]}=\frac{1}{\beta-1}$); or
	\item be the same as a previous gap (implying, by Lemma~\ref{lem:EventuallyPeriodicGaps}, an eventually-periodic representation); or 
	\item will be such that $\hat{c}$ is not divisible by $q$ (implying, by the previous claim, that there is no eventually-periodic representation).
\end{itemize}
\end{enumerate}

Indeed: 
\begin{enumerate}
\item Assume that $c$ is not divisible by $q$. Consider the prime factorization $q=f_1^{e_1} \cdot f_2^{e_2} \cdots f_n^{e_n}$ of $q$. Then, by the assumption, there is some $1 \leq i \leq n$, such that $c$ is not divisible by $f_i^{e_i}$. Thus, the numerator of $g'$ (which is $pc - qmd$) is not divisible by $q$, as (i) $pc$ is not divisible by $f_i^{e_i}$, since $p$ is co-prime with $q$ and $c$ is not divisible by $f_i^{e_i}$, and (ii) $qmd$ is divisible by $f_i^{e_i}$, since $q$ is. Hence, the exponent of $f_i$ in the denominator of $g'$ is bigger than its exponent in $g$. Therefore, by induction, the exponent of $f_i$ in the denominators of the gaps monotonically increases, precluding the possibility of two equivalent gaps.
\item Assume that $c$ is divisible by $q$. Then, the numerator $pc - qmd$ in the calculation of $g'$ is divisible by $q$, meaning that $g'=\frac{c'}{d}$, for some integer $c'$. Thus, as long as the numerators of the gaps are divisible by $q$, the denominator cannot grow. Hence, within $d$ steps of exploration, we must reach a gap that satisfies one of the following properties: (i) Its numerator exceeds $d$, and therefore it is bigger than \Do, which is bigger than $\frac{1}{\beta-1}$; or (ii) It is the same as a previous gap; or (iii) Its numerator is not divisible by $q$.
\end{enumerate}

As for the space complexity, by the above, all the gaps are in the form of $\frac{x}{d}$, where $x \leq d$. Thus, the procedure can be done using a polynomial space.
\end{proof}

Using Lemma~\ref{lem:EP}, we get the decidability of both \TDS and \CTDS for eventually-periodic sequences. The following theorem only holds for a discount factor $\lambda < \frac{1}{2}$, and will be generalized in Theorem~\ref{thm:GTDSEventuallyPeriodic} to hold for an arbitrary discount factor.

\begin{restatable}{theorem}{TdsEP}
\label{thm:TdsEP} 
Consider an instance $\P$ of \CTDS (or \TDS) with a discount factor $\lambda < \frac{1}{2}$. Then the problem of whether $\P$ has an eventually-periodic solution is decidable in PSPACE.
\end{restatable}
\begin{proof}
The \TDS case is a corollary of Theorem~\ref{thm:TdsToTdszo} and Lemma~\ref{lem:EP}. 

As for the \CTDS case, note that the proof of Lemma~\ref{lem:EP} also generates the (unique) representation, when it exists. 
The representation is given in the form of $uv^{\omega}$ with $u,v \in \{0,1\}^{*}$.
Hence, for checking whether there is a representation that satisfies an $\omega$-regular expression $e$, one can check in PSPACE whether the generated representation belongs to the language of $e$ \cite{GPVW95,Var96}.
\end{proof}

Another corollary of Lemma~\ref{lem:EP} is the case when the discount factor $\lambda$ is of the form $\frac{1}{n}$, where $n$ is a natural number.

\begin{restatable}{theorem}{IntegralDenominator}
\label{thm:IntegralDenominator}
\TDS (resp.\ \CTDS) is in PSPACE for every discount factor $\lambda$ of the form $\frac{1}{n}$ where $n$ is a natural number.
\end{restatable}
\begin{proof}
Reducing \TDS to \TDSzo (Theorem~\ref{thm:TdsToTdszo}), we consider a representation of a target rational number in base $\frac{1}{\lambda} = n$, which is a natural number.
As a representation of a rational number in an integral base is always eventually periodic \cite{Ren57}, we can decide the existence of a representation using Lemma~\ref{lem:EP}.

As for \CTDS, we use an argument analogous to the one provided for Theorem~\ref{thm:TdsEP}.
\end{proof}

\Paragraph{Growth rate of blocks of the same digit}  By Lemma~\ref{lem:EP}, given a rational number, we know whether it has an eventually-periodic representation with only \Dz's and \Do's. However, when it does not have it, we do not know whether it has a non-eventually-periodic representation, or not. We do not even know whether, for the specific base $\frac{5}{2}$, there exists a rational number with a non-eventually-periodic representation that only uses \Dz's and \Do's. 

One can ask a related question from the opposite direction -- given a non-eventually-periodic representation, does it represent a rational number? For example, consider the number $0.101000001...$ in base $\frac{5}{2}$, where \Do's appear exactly at positions $3^n$. Is it a rational number?

In some cases, such as the example above, we know that the answer is negative. We show below that in a representation of a rational number, the number of consecutive \Dz's and consecutive \Do's is bounded by an exponential rate.

\begin{restatable}{lemmaStatement}{GrowthRate}
\label{lem:GrowthRate}
Consider a representation of a rational number $\frac{c}{d}$ in base $\frac{p}{q} > 2$, using only \Dz's and \Do's. Then, i) The first \Do occurs within the first $\log d$ positions, and ii) if there is a \Do in the $n$-th position, then the next \Do occurs at position at most $n(\log q + 1)+\log d$.
\end{restatable}
\begin{proof}
We consider the greedy exploration, and analyze the minimal positive gap that might occur after using a \Do.  Note that having a gap $g=\frac{a}{b}$, and adding a digit $m$, the next gap is $g'= \frac{p}{q} \cdot \frac{a}{b} - m  = \frac{pa - qmb}{qb}$. Thus, the denominator of the gap at position $n$ is smaller than or equal to $d q^n$, while the numerator is, obviously, bigger than or equal to \Do. Hence, the gap is bigger than or equal to $\frac{1}{d q^n}$. 

Now, the maximal gap that has a representation starting with \Dz is smaller than $0.1^\omega_{[\frac{p}{q}]} = \frac{q}{p}$. Since the gap is multiplied by $\frac{p}{q}$ after every \Dz digit, it means that the next \Do digit must occur within $i$ steps, such that $\frac{1}{d q^n} \cdot (\frac{p}{q})^i \geq \frac{q}{p}$. Hence, the next \Do must occur within $\log_{\frac{p}{q}} (\frac{q}{p} \cdot d q^n) < \log  (\frac{q}{p} \cdot d q^n)  < \log  (d q^n)  = n \log q +\log d$ steps.
\end{proof}

We shall call a maximal subword consisting of the same letter a \emph{block}.

\begin{restatable}{theorem}{LongBlocks}
Let $\beta \in \Rat$ and $w \in \{ \Dz, \Do\}^{\omega}$. 
The number $(\Dz.w)_{[\beta]}$ is rational only if 
for every $n$, the block of $w$ starting at the $n$th position is linearly bounded in $n$.
\end{restatable}
\begin{proof}
Let $\beta=\frac{p}{q}$ and $w \in \{ \Dz, \Do\}^{\omega}$.
Assume that  $(\Dz.w)_{[\frac{p}{q}]} = \frac{c}{d}$, where $c,d\in\Nat$.
Lemma~\ref{lem:GrowthRate} implies that for every $n$,
the length of a block of \Dz's starting at the $n$th position is linearly bounded.

As for blocks of \Do's, consider the value $t' = (\Dz.\Do^{\omega})_{[\frac{p}{q}]} - (\Dz.w)_{[\frac{p}{q}]} =
\frac{1}{\beta - 1} - \frac{c}{d}$. 
Note that $t'$ is a rational, and its representation is obtained from $w$ by swapping all $\Dz$'s with $\Do$'s and vice versa.
Thus, by applying Lemma~\ref{lem:GrowthRate} on $t'$, for every $n$, the length of a block of \Do's starting at the $n$th position of $w$ is linearly bounded.
\end{proof}

\section{Generalized \TDS\ -- Analysis and Results}\label{sec:GtdsResults}
We extend below the results that were shown for \TDS and \CTDS to the generalized versions \GTDS and \CGTDS, in which there are arbitrarily many weights. The main difficulty here is that there is no analogous to the combination of Theorem~\ref{thm:BiggerThanHalf} and Lemma~\ref{lem:UniqueZO}, stating that a relevant representation is either guaranteed, or else guaranteed to be unique, if it exists. Indeed, allowing enough digits, a number might not have, or may have many, and even infinitely many, different representations in a nonintegral base \cite{EJK94,GS01}. We overcome the problem by using boundedly-branching trees, rather than unique words, and applying K\"onig's Lemma, as well as a few other observations.

We begin with establishing a normal form of \GTDS, generalizing \TDSzo, in which the weights are natural numbers
and the least weight is $0$. 

\begin{restatable}{theorem}{GTdsToGTdszo}
\label{thm:GTdsToGTdszo}
\CGTDS (or \GTDS) polynomially reduces to \CGTDS (or \GTDS) with weights from $\Nat$, where the least weight is $0$.
\end{restatable}
\begin{proof}
Consider an instance $\P$ of \CGTDS with a discount factor $\lambda$, a target $t$, weights $\Dm{1} < \ldots < \Dm{k} \in \Rat$, and a constraint $e$. 
Let $M$ be the least common denominator of the weights. 
We define a \CGTDS $\P'$ with a discount factor $\lambda$, a target $t'=t\cdot M$, weights $\Dm{1}'=\Dm{1}\cdot M, \ldots, \Dm{k}'=\Dm{k}\cdot M $, and a constraint $e'$ that is the same as $e$, just referring to $\Dm{i}'$ instead of $\Dm{i}$, for all $1\leq i \leq k$. Note that all the weights in $\P'$ are integers. Now, observe that  $\P$ has a solution (resp.\ eventually-periodic solution) iff $\P'$ has a solution (resp.\ eventually-periodic solution), since for every sequence $w$, we have $\sum_{i=0}^{|w|} M \cdot w(i)  \lambda^i = M \cdot \sum_{i=0}^{|w|}  w(i)  \lambda^i$.
Finally, for having the least weight $0$, we subtract $\Dm{1}\cdot M$ from all weights, having the \CGTDS instance $\P''$ with a discount factor $\lambda$, a target $t''=(t\cdot M) - \frac{\Dm{1}\cdot M}{1 - \lambda}$, weights $\Dm{1}''=0, \Dm{2}''=(\Dm{2}-\Dm{1}) \cdot M, \ldots, \Dm{k}''=(\Dm{k}-\Dm{1})\cdot M $, and a constraint $e''$ that is the same as $e$, just referring to $\Dm{i}''$ instead of $\Dm{i}$, for all $1\leq i \leq k$.
Notice that $\P''$ is in normal form and equivalent to $\P$ by an argument analogous to part I in the proof of Theorem~\ref{thm:TdsToTdszo}.
%See Example E at the end of the section.
\end{proof}

\begin{exampleE}
Consider a \GTDS instancde with $\lambda = \frac{2}{3}$, weights $\frac{-1}{2}, 0, \frac{1}{2}$ and the target $\frac{-3}{10}$. 
It reduces to a \GTDS instance with the same $\lambda = \frac{2}{3}$, natural weights $0,1,2$ and the target $\frac{12}{5}$.
\end{exampleE}

In the remaining part of this section we will assume that all instances of \GTDS and \CGTDS are in normal form. This allows us to consider the problem in the setting of $\beta$-expansions.

\begin{proposition}
A \GTDS with a discount factor $\lambda$, a target value $t$, and weights $\Dm{1}=0  < \Dm{2} < \ldots < \Dm{k} \in \Nat$ has a solution iff $t$ has a representation in base $\frac{1}{\lambda}$, using only the digits $\Dm{1}, \ldots, \Dm{k}$.
\end{proposition}

With \TDS, Theorem~\ref{thm:BiggerThanHalf} shows that all relevant target numbers have a solution when the discount factor is equal to or bigger than half. This obviously also holds for \GTDS.  Further, having more weights, there are additional cases that guarantee a solution, as characterized below.

\begin{restatable}{theorem}{UnfailingGTDS}
Consider a \GTDS instance $\P$ with a discount factor $\lambda$ and weights $\Dm{1}=0  < \Dm{2} < \ldots < \Dm{k} \in \Nat$.
Then $\P$ has a solution for every target $t \in [ (\Dz.\Dz^{\omega})_{[\frac{1}{\lambda}]},  (\Dz.\Dm{k}^{\omega})_{[\frac{1}{\lambda}]}]$ iff
for every $i \in \{1,\ldots, k-1\}$ we have 
$\Dm{i+1} - \Dm{i} \leq  (\Dz.\Dm{k}^{\omega})_{[\frac{1}{\lambda}]}$. 
\end{restatable}
\begin{proof}\

$\Rightarrow$: 
Suppose there exists $i \in \{1,\ldots, k-1\}$, such that $\Dm{i+1} - \Dm{i} >  (\Dz.\Dm{k}^{\omega})_{[\frac{1}{\lambda}]}$. 
Then, $(\Dz.\Dm{i}\Dm{k}^{\omega})_{[\frac{1}{\lambda}]} < (\Dz.\Dm{i+1})_{[\frac{1}{\lambda}]}$. Accordingly, all the numbers that are bigger than $(\Dz.\Dm{i}\Dm{k}^{\omega})_{[\frac{1}{\lambda}]}$ and smaller than $(\Dz.\Dm{i+1})_{[\frac{1}{\lambda}]}$ have no representation in base $\frac{1}{\lambda}$ with the digits $\Dz, \Dm{2}, \ldots, \Dm{k}$.

$\Leftarrow$: Consider an execution of the greedy exploration with gaps on the target number $t$. 
Note that having a gap $g$ and adding a digit $m$, the next gap is $g'= \frac{1}{\lambda} \cdot g - m$.
It suffices to show that the following invariant holds, stating that the gaps are always not too small and not too big:

\noindent (*)~if a gap $g \in [0,  (\Dz.\Dm{k}^{\omega})_{[\frac{1}{\lambda}]}]$, 
and $g' = \frac{1}{\lambda} g - \Dm{i}$, where $\Dm{i}$ is the maximal among $\Dm{1}, \ldots, \Dm{k}$ for which $g' \geq 0$, 
then $g' \in [0,  (\Dz.\Dm{k}^{\omega})_{[\frac{1}{\lambda}]}]$. 

To prove (*) we consider two cases. If $i < k$, then by the maximality of $i$,
$\frac{1}{\lambda} g - \Dm{i+1} < 0$.
It follows that $g' < \Dm{i+1} - \Dm{i} < (\Dz.\Dm{k}^{\omega})_{[\frac{1}{\lambda}]}$.
Therefore, $g' \in [0,  (\Dz.\Dm{k}^{\omega})_{[\frac{1}{\lambda}]}]$.
Otherwise, if $i = k$, $g \leq  (\Dz.\Dm{k}^{\omega})_{[\frac{1}{\lambda}]}$ implies 
$g'  = \frac{1}{\lambda} g - \Dm{k} \leq  (\Dz.\Dm{k}^{\omega})_{[\frac{1}{\lambda}]}$.
\end{proof}

For \TDS, Lemma~\ref{lem:UniqueZO} shows that when the discount factor does not guarantee a solution it guarantees that the solution, if exists, is unique. This no longer holds for \GTDS. Yet, as in the case of \TDS (Proposition~\ref{prop:TDSisCoRE}), \GTDS is also co-recursively-enumerable.

\begin{restatable}{theorem}{GTDSisCoRE}
\label{thm:GTDSisCoRE}
\GTDS is co-recursively-enumerable.	
\end{restatable}
\begin{proof}
Assume that a \GTDS instance $\P$  has no solution and consider a tree consisting of all runs of the nondeterministic exploration. 
Infinite paths in that tree correspond to solutions to $\P$, therefore each failing path is finite. 
Since the exploration tree has a finite degree and all of its paths are failing, K\"onig's Lemma implies that it is finite. 
Therefore, exhausting all the possibilities of the nondeterministic exploration on $\P$, we are guaranteed to stop with the result that $\P$ has no solution.
\end{proof}

We continue with solving \GTDS w.r.t.\  eventually-periodic sequences. The following lemma generalizes Lemma~\ref{lem:EventuallyPeriodicGaps}.

\begin{restatable}{lemmaStatement}{GeneralizedEventuallyPeriodicGaps}
\label{lem:GeneralizedEventuallyPeriodicGaps}
Consider a rational number $t$, a rational base $\beta > 1$, and digits $\Dm{1}, \ldots, \Dm{k} \in \Nat$, for $k\in\Nat$. Then, 
$t$ has an eventually-periodic $\beta$-representation iff there are finitely many different gaps in some run of the corresponding nondeterministic exploration of $t$.
\end{restatable}
\begin{proof}
Analogous to the proof of Lemma \ref{lem:EventuallyPeriodicGaps}. 
Note that in the proof of Lemma \ref{lem:EventuallyPeriodicGaps} we did not stipulate that $m \in \{0,1\}$ and hence the proof works for all $m \in \Nat$.
\end{proof}

We are now in place to provide the result on eventually-periodic sequences.
\begin{restatable}{theorem}{GTDSEventuallyPeriodic}
\label{thm:GTDSEventuallyPeriodic}
Consider an instance $\P$ of \CGTDS (or \GTDS) with a rational discount factor $\lambda$, a rational target $t=\frac{c}{d}$,  weights $\Dm{1} < \Dm{2} < \ldots < \Dm{k} \in \Nat$, and an $\omega$-regular constraint $e$. Then the problem of whether $\P$ has an eventually-periodic solution is decidable in PSPACE.
Moreover, if there is an eventually-periodic solution to $\P$, then there is also a solution $u v^{\omega}$ with $|u| + |v| \leq 2\cdot d \cdot |e| \cdot \frac{\lambda \Dm{k}}{1 - \lambda}$.
\end{restatable}
\begin{proof}
In the case of $\GTDS$, the decision procedure is analogous to the proof of Lemma~\ref{lem:EP}, with the following differences: i) Rather than considering the presentation generated by the greedy exploration, we consider some representation generated by the nondeterministic exploration; ii) The first gap is $t \cdot \lambda$ rather than $t$, i.e., we multiply the target by $\lambda$ 
to make it consistent with the calculations made for \TDSzo, where the summation starts with $i=1$; iii) The upper bound on relevant gaps is $\frac{\lambda \Dm{k}}{1 - \lambda}$; 
and iv) A decision is nondeterministically reached within $\lfloor d \cdot \frac{\lambda \Dm{k}}{1 - \lambda} \rfloor$ steps.
Hence, the procedure runs in nondeterministic polynomial space in the binary representations of $\lambda, c, d, \Dm{1}, \ldots, \Dm{k}$. By 
Savitch's Theorem, from which we have PSPACE = NPSPACE, we get the PSPACE complexity.

As for the generalization to a \CGTDS instance $\P'$, having a \GTDS problem $\P$ as above and an $\omega$-regular constraint $e$, we cannot assume, as opposed to the case of Theorem~\ref{thm:TdsEP}, a unique solution to $\P$ -- there are possibly infinitely many different solutions, some of which satisfy $e$ and some do not. Thus, we cannot just check the membership of some solution in the language of $e$, as is done in the proof of Theorem~\ref{thm:TdsEP}.

Yet, we can have a nondeterministic B\"uchi automaton $\A$ (which is even, after the removal of dead-end states, a safety automaton, namely a B\"uchi automaton all of whose states are accepting) that recognizes all eventually-periodic solutions of $\P$, as well as some non-eventually periodic solutions, while possibly not all of them: The states of $\A$ are all the gaps along the exploration that have denominator $d$ and value up to $\frac{\lambda \Dm{k}}{1 - \lambda}$ -- once the gap is bigger than $\frac{\lambda \Dm{k}}{1 - \lambda}$, it cannot lead to a solution, and once the gap's denominator is different from $d$, it cannot lead to an eventually-periodic solution (as shown in the proof of Lemma~\ref{lem:EP}). Hence, $\A$ has up to $\lfloor d \cdot \frac{\lambda \Dm{k}}{1 - \lambda} \rfloor$ states.

Now, let $\B$ be a nondeterministic B\"uchi automaton that recognizes the language of $e$ (which is of size up to twice the length of $e$), and $\C$ a nondeterministic B\"uchi automaton for the intersection of $\A$ and $\B$ (which is of size up to $|\A|\cdot|\B|$). We can check in time linear in the size of $\C$ whether it is empty, and conclude that there is an eventually-periodic solution to $\P'$ iff $\C$ is not empty. Indeed, even though $\A$ also allows for non-eventually-periodic solutions, due to the $\omega$-regularity of $\C$, it is guaranteed that if $\C$ is not empty, then there is an eventually-periodic word in its language, whose prefix and period together are of length at most $|\C|$.
\end{proof}

\begin{exampleE}[for the proof of Theorem~\ref{thm:GTDSEventuallyPeriodic}]
Consider the \CGTDS instance with a discount factor $\lambda=\frac{2}{3}$, a target $t=\frac{12}{5}$,
weights $a=0$, $b=1$, and $c=2$, and an $\omega$-regular constraint 
$e = \Sigma^*(\Sigma a)^\omega$, namely eventually the letter a occurs at every second position.
We follow the proof of Theorem~\ref{thm:GTDSEventuallyPeriodic} for checking whether it has an eventually-periodic solution. 

We start with building the safety automaton $\A$ of the relevant gaps along the exploration of $t$ (see Figure~\ref{fig:EP}):
The first gap is $t \cdot \lambda = \frac{8}{5}$. The minimal relevant gap is $0$, the maximal relevant gap is 
$\frac{\lambda c}{1-\lambda} = \frac{\frac{2}{3} \cdot 2}{\frac{1}{3}} = 4$, and all relevant gaps have denominator $5$ (as otherwise they cannot be continued to an eventually-periodic solution). 
Moreover, following the argument in the proof of Lemma~\ref{lem:EP}, 
we can omit further exploration of gaps with odd numerator  
as all its continuations will have denominators growing to infinity. 
\begin{figure}
\begin{tikzpicture}[semithick,initial text=, every initial by arrow/.style={|->},state/.style={circle, draw, thick, minimum size=0.9cm,inner sep=0,outer sep=0}]
\tikzset{good/.style={accepting}}
\tikzset{tooBig/.style={fill=red!20}}
\tikzset{odd/.style={fill=orange!20}}

\node[anchor=west]  (A) at (-1.0,4.5) {\Large$\A$};
\node[anchor=west]  (A_TooBig) at (-1.0,3.9) {{\color{red} Red gaps}: Too big/small};
\node[anchor=west]  (A_Odd) at (-1.0,3.5) {{\color{orange} Orange gaps}: Odd numerator};

\node[state,initial, good] (Eps) at (0,0) {$\frac{8}{5}$}; 
\foreach \n/\x/\y\l/\f in {
						a/2/2/{\frac{12}{5}}/good,
						aa/4/2.75/{\frac{18}{5}}/good,
						ab/4/1.25/{\frac{13}{5}}/odd,
						aaa/6/4.25/{\frac{27}{5}}/tooBig,
						aab/6/2.75/{\frac{22}{5}}/tooBig,
						aac/6/1.25/{\frac{17}{5}}/odd,
						b/2/0/{\frac{7}{5}}/odd,
						c/2/-2/{\frac{2}{5}}/good,
                        ca/4/-0.5/{\frac{3}{5}}/odd,
                        cb/4/-2/{\frac{-2}{5}}/tooBig,
                        cc/4/-3.5/{\frac{-7}{5}}/tooBig}
{
	\node[state,\f] (\n) at (\x,\y) {$\l$}; 
}

\foreach \from/\to/\lab in {Eps/a/a,
                              Eps/b/b, 
                              Eps/c/c,
                              a/aa/a,
                              aa/aaa/a,
                              aa/aab/b,
                              aa/aac/c,
                              a/ab/b,
                              Eps/c/c,
                              c/ca/a,
                              c/cb/b,
                              c/cc/c}
{
	\draw[->] (\from) to node[above]{\lab} (\to) ;
}

\draw[->,bend right] (a) to node[above]{c} (Eps) ;

\begin{scope}[xshift=9cm,yshift=2cm]
	
\node[anchor=west]  (A) at (-0.75,2.5) {\Large$\A$};
\node[anchor=west]  (A_TooBig) at (-0.75,1.9) {After removing dead-end states};
	
\node[state, initial above,accepting] (q0) at (0,0) {$\frac{8}{5}$};
\node[state,accepting] (q1) at (0,-3) {$\frac{12}{5}$};

\draw[->,bend left] (q0) to node[left]{a} (q1) ;
\draw[->,bend left] (q1) to node[right]{c} (q0) ;

\end{scope}

\begin{scope}[xshift=0cm,yshift=-6cm]
\node[anchor=west]  (B) at (-1.0,1.8) {\Large$\B$};

\node[state,initial] (q0) at (0,0) {$q_0$};
\node[state] (q1) at (2,0) {$q_1$};
\node[state,accepting] (q2) at (4,0) {$q_2$};

\draw[->, loop above] (q0) to node[above]{a,b,c} (q0) ;
\draw[->,bend left] (q0) to node[above]{a,b,c} (q1) ;
\draw[->,bend left] (q1) to node[above]{a,b,c} (q2) ;
\draw[->,bend left] (q2) to node[below]{a} (q1) ;
\end{scope}

\begin{scope}[xshift=7cm,yshift=-5cm]
\node[anchor=west]  (C) at (-1.0,0.8) {\Large$\C$};
	
\node[state,initial] (q0A) at (0,0) {$q_0,\frac{8}{5}$};
\node[state] (q0B) at (0,-2) {$q_0,\frac{12}{5}$};
\node[state] (q1A) at (2.5,0) {$q_1,\frac{8}{5}$};
\node[state] (q1B) at (2.5,-2) {$q_1,\frac{12}{5}$};
\node[state,accepting] (q2A) at (5,0) {$q_2,\frac{8}{5}$};
\node[state,accepting] (q2B) at (5,-2) {$q_2,\frac{12}{5}$};

	\draw[->,bend right] (q0A) to node[left]{a} (q0B);
	\draw[->,bend right] (q0B) to node[left]{c} (q0A);
	
	\draw[->, bend left=10] (q0B) to node[above, yshift=-10pt, xshift=-10pt ] {c} (q1A);
	\draw[->,bend left=10] (q0A) to node[below, yshift=0pt,xshift=0pt ]{a} (q1B);
	
	\draw[->,bend right=20] (q1B) to node[below]{c} (q2A);
	\draw[->, bend left] (q1A) to node[above,yshift=5pt,xshift=-10pt]{a} (q2B);
	
	\draw[->,bend right=20] (q2A) to node[above,yshift=-5pt,xshift=-5pt]{a} (q1B);
	
\end{scope}

\end{tikzpicture}

\caption{The automata of the example for the proof of Theorem~\ref{thm:GTDSEventuallyPeriodic}: The safety automaton $\A$ for the exploration of the relevant gaps (stopping the exploration for too big/small gaps and for gaps with odd numerators), the B\"uchi automaton $\B$ expressing the $\omega$-regular constraint, and the B\"uchi automaton $\C$ for the intersection of $\A$ and $\B$.}
\label{fig:EP}

\end{figure}
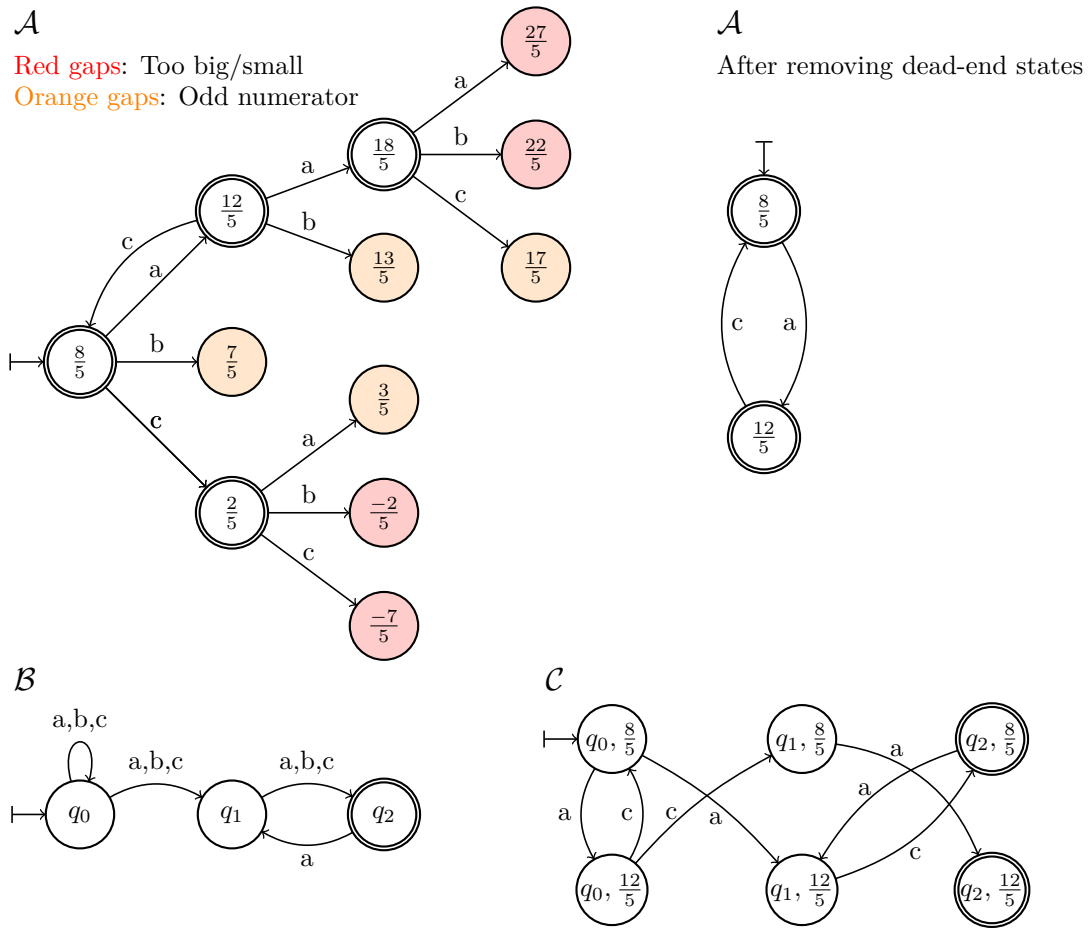

Next, we build  the B\"uchi automaton $\B$ expressing the constraint $e$, and the B\"uchi automaton $\C$ for the intersection of $\A$ and $\B$.
Notice that the intersection of a safety automaton and a B\"uchi automaton is the B\"uchi automaton 
derived from their product, in which we can also ignore dead-end states.

Finally, we check for the nonemptiness of $\C$: In this case, it is not empty, having the word $w=a(ca)^\omega$, which is an eventually-periodic solution to the considered \CGTDS instance. 
Observe that indeed, $(\Dz.(20)^{\omega})_{[\frac{1}{\lambda}]}$ is the required target $\frac{12}{5}$.
\qed
\end{exampleE}

Notice that in contrast to the above example, eventually-periodic solutions to $\CGTDS$ need not be unique. For example, 
consider $\lambda = \frac{1}{2}$ and weights $0,1,2$. Then, 
$ 
\Dz.(10)^{\omega})_{[\frac{1}{\lambda}]} = \Dz.(02)^{\omega})_{[\frac{1}{\lambda}]}
$.

As a corollary of Lemma~\ref{lem:GeneralizedEventuallyPeriodicGaps}, we get the decidability for the case that the discount factor is of the form $\frac{1}{n}$ for $n\in\Nat$.

\begin{restatable}{theorem}{GTDSIntegralDenominator}
\CGTDS (or \GTDS) is in PSPACE for every discount factor $\lambda$ of the form $\frac{1}{n}$ where $n$ is a natural number.
\end{restatable}
\begin{proof}
Consider an instance of $\GTDS$ with a discount factor $\lambda = \frac{1}{n}$ and a target 
number $t = \frac{c}{d}$.
Observe that a nondeterministic exploration can only visit gaps whose denominator is $d$.
Indeed, if a gap $g = \frac{a}{d}$ then the next gap, after using a digit $m$, is $g' = \frac{a n}{d} - m$, also having the denominator $d$.

Therefore, analogously to the arguments given in the proof of Theorem~\ref{thm:GTDSEventuallyPeriodic}, we get a bound on the gap numerators, implying decidability in PSPACE of both the \GTDS and the \CGTDS instances.
\end{proof}

We conclude the section with the result on the decidability of the generalized target discounted-sum problem over finite words.

\begin{restatable}{theorem}{GTDSFinite}
\label{thm:GTDSFinite}
Consider an instance $\P$ of \CGTDSF (or \GTDSF) with a regular constraint $e$, and a rational discount factor $\lambda=\frac{p}{q}$, a target $t=\frac{c}{d}$, and weights $\Dm{1} < \Dm{2} < \ldots < \Dm{k} \in \Rat$ with a common denominator $m\in\Nat$. Then the problem of whether $\P$ has an eventually-periodic solution is decidable in PSPACE.
Moreover, if there is a solution to $\P$, then there is also a solution of length up to $4\cdot d(q-p) \cdot (|e|+2) \cdot \max(|\Dm{k}|,|\Dm{1}|) \cdot \frac{\lambda\cdot m}{1-\lambda}$.

\end{restatable}
\begin{proof}
By the proof of Theorem~\ref{t:FiniteTdsToCgtds}, an instance $\P$ of \CGTDSF can be reduced to an instance $\P'$ of \CGTDS, having the same input as $\P$, except for possibly adding a $0$ weight, and concatenating to the constraint $e$ the suffix $0^\omega$ of length 2. 

By the proof of Theorem~\ref{thm:GTdsToGTdszo}, an instance $\P'$ of \CGTDS can be reduced to an instance $\P''$ of \CGTDS in normal form, where the new target is 
$t''= (t\cdot m) - \frac{\Dm{1}\cdot m}{1 - \lambda} = 
\frac{m (c(q-p) - d(\Dm{1}\cdot q))}{d(q-p)}$ 
and the new highest weight is $\Dm{k}''=(\Dm{k}-\Dm{1})m$ if $\Dm{1}<0$ and $\Dm{k}>0$, or $\Dm{k}''=|\Dm{k}|m$ if $\Dm{1}\geq 0$ and $\Dm{k}>0$, or $\Dm{k}''=|\Dm{1}|m$ if $\Dm{1}<0$ and $\Dm{k}\leq0$, implying that in any case $\Dm{k}''\leq 2\max(|\Dm{k}|,|\Dm{1}|) \cdot m$.

Hence, the PSPACE decidability and the bound on a shortest witness directly follow from Theorem~\ref{thm:GTDSEventuallyPeriodic}.
%See Example E at the end of the section.
\end{proof}

%\begin{exampleE}
%Consider a \GTDS instance with the same $\lambda = \frac{2}{3}$, natural weights $0,1,2$ and 
%the target $\frac{40}{27}$.
%The greedy exploration works as follows.
%The initial gap is $\frac{80}{81}$ and the numerator $2$ of $\lambda$ divides $80$. 
%The only possible weight to keep the gap non-negative is $0$.
%The next gap is $\frac{40}{27}$ and again $2$ divides $40$. 
%The maximal weight to keep the gap non-negative is $2$ as $\frac{40}{27} \cdot \frac{3}{2} = \frac{20}{9}$. 
%The next gap is $\frac{2}{9}$. For this gap, the only possible weight to keep gaps non-negative is $0$, but 
%then the next gap is $\frac{1}{3}$ and its numerator is odd.
%
%However, observe that $\frac{40}{27}$ can be represented with the expansion $(\Dz.0022)_{[\frac{1}{\lambda}]}$:
%\[
%\frac{40}{27} =  2 \cdot \left(\frac{2}{3}\right)^3 + 2 \cdot \left(\frac{2}{3}\right)^4 = 
%\frac{16}{27} + \frac{32}{81}
%\]
%the sequence of gaps is as follows: $\frac{80}{81},\frac{40}{27},\frac{20}{9},\frac{4}{3}$ and finally $0$.
%\end{exampleE}

\section{Results on Discounted-Sum Automata}\label{sec:DSA}

In this section, we establish the connection between \TDS and discounted-sum automata, and use our results about \TDS for solving some of the latter's open problems. In particular, we solve the exact-value problem for nondeterministic automata over finite words and the universality and inclusion problems for functional automata.

We start with the definitions of discounted-sum automata and their related problems.

\Paragraph{Discounted-sum automata} 
A \emph{discounted-sum automaton} (\DSA) is a tuple $\A = \tuple{\Sigma, Q, q_{in}, Q_F, \delta, \gamma, \lambda}$ over a finite alphabet $\Sigma$, with a finite set of states $Q$, an initial state $q_{in}\in Q$, a set of accepting states $Q_F \subseteq Q$, a transition function $\delta \subseteq Q \times \Sigma \times Q$, a weight function $\gamma: \delta\to \Rat$, and a rational discount factor $0 < \lambda < 1$.

A run of an automaton on a word $w = \sigma_1 \sigma_2 \ldots$ is a sequence of states and letters, $q_0, \sigma_1, q_1, \sigma_2, q_2, \ldots$, such that $q_0=q_{in}$ and for every $i$, $(q_i,\sigma_{i+1},q_{i+1})\in\delta$. The length of a run $r$, denoted by $|r|$, is $n$ for a finite run $r = q_0, \sigma_1, q_1, \ldots, \sigma_n, q_n$, and $\infty$ for an infinite run. 
A finite run of an automaton $\A$ is \emph{accepting} if the last state of $r$ belongs to $Q_F$.
In the infinite case, we assume that every run is \emph{accepting} (and $Q_F$ is irrelevant).

The value of a run $r$ is $\gamma(r) = \sum_{i=0}^{|r|-1} \lambda^i \cdot {\gamma(q_i,\sigma_{i+1},q_{i+1})}$. The value of a word $w$ (finite or infinite) is $\A(w) = \inf \{\gamma(r) \ST \mbox{$r$ is an accepting run of $\A$ on $w$}
\}$. 

A \DSA $\A$ over finite words is said to be \emph{functional} if for every word $w$, 
all accepting runs of $\A$ on $w$ have the same value \cite{functional}. (Notice that functional automata are less general than nondeterministic ones, while more general than unambiguous and deterministic ones.)

\Paragraph{Decision problems}
Given \DSAs $\A$ and $\B$ and a value $t\in\Rat$, 
\begin{itemize}
\item the \emph{exact-value} problems asks whether there exists a word $w$ such that $\A(w) = t$,
\item the \emph{$<$-universality} (resp. \emph{$\leq$-universality}) problem asks whether for every word $w$ we have  $\A(w) < t$ (resp., $\A(w) \leq t$).
\item the \emph{$<$-inclusion} (resp. \emph{$\leq$-inclusion}) problem asks whether for every word $w$ we have  $\A(w) < \B(w)$ (resp., $\A(w) \leq \B(w)$).
\end{itemize}

Next, we establish the connection between the target discounted-sum problem and the above decision problems. 

\Paragraph{Results for finite words}
Our techniques for resolving the target discounted-sum problem over finite words directly relate to the exact-value problem:

\begin{restatable}{theorem}{ExactValue}
\label{thm:ExactValue}
The exact-value problem for functional (as well as unambiguous and deterministic) discounted-sum automata is decidable in PSPACE.
\end{restatable}
\begin{proof}
Consider a functional \DSA $\A$ with discount factor $\lambda$, a target value $t\in\Rat$, and the problem $\P$ of whether there exists a finite word $u$ accepted by $\A$ with the value $t$.

Notice that by the functionality of $\A$, such a word $u$ exists iff there exists some accepting path of $\A$ (i.e., a path that starts from the initial state and ends in an accepting state) with value $t$. Indeed, if there is no such path then obviously there is no such word, while if there is such a path, over some word $u$, then since all accepting paths over $u$ have the same value, we get that the value of $\A$ on $u$ is $t$.

Let $\A'$ be the deterministic finite-state automaton (DFA) that is derived from $\A$ by ignoring the alphabet letters and weights, and setting the alphabet letter over each transition to be the name of the transition. Notice that the language of $\A'$ consists of all the accepting paths of $\A$.

Now, consider the \CGTDSF $\P'$ with discount factor $\lambda$, target value $t$, and regular constraint $\A'$, whose weight-letters are the transitions names in $\A$ and the weight of each letter is the weight of the transition in $\A$. (Notice that there may be several weight-letters with the same weight.)

Observe that every solution to $\P'$ provides a path of $\A$ with the value $t$ and vice versa.
By Theorem~\ref{thm:GTDSFinite}, if there is a solution to $\P'$ then there is a solution of length polynomial in the input. Hence, we can guess in polynomial space a solution, getting that the problem is in NPSPACE, and by Savitch's Theorem also in PSPACE.
\end{proof}

Notice that the proof of Theorem~\ref{thm:ExactValue} does not hold for general nondeterministic discounted-sum automata, since one cannot consider only the existence of an accepting path with the required value, but should rather consider all possible paths of the automaton on each word.

Continuing with functional automata, it is shown in \cite{functional} that the non-strict versions of the inclusion and universality problems are decidable in PTIME. They leave the strict versions of these problems as an open question. Our result about the exact-value problem provides an immediate solution to these open problems.

\begin{restatable}{theorem}{FuncInc}
\label{thm:FunctionalInclusion}
The inclusion and universality problems of functional discounted-sum automata are decidable in PSPACE.
\end{restatable}
\begin{proof}
The $\leq$-inclusion and $\leq$-universality problems are shown in \cite{functional} to be in PTIME. 
A solution to the $<$-universality problem then directly follows from Theorem~\ref{thm:ExactValue}: Given a functional \DSA $\A$ and a threshold value $t$, for every finite word $w$ $\A(w)<t$ iff for every finite word $w$ we have $\A(w)\leq t$, and there is no word $w$ such that $\A(w)=t$.

As for the $<-inclusion$ problem, given functional \DSAs $\A$ and $\B$, one can construct in PSPACE a functional \DSA $\C$, such that for every word $w$ we have $\C(w)=\A(w)-\B(w)$ \cite{functional}. Then, the $<-inclusion$ problem for $\A$ and $\B$ reduces to the $<$-universality problem for $\C$ and the threshold $0$.
\end{proof}

\Paragraph{Results for infinite words}

In the case of infinite words, \TDS reduces to the $<$-universality problem, which in turn reduces to the $<$-inclusion problem. As for the other direction, there is a partial implication: given a \DSA $\A$ and a threshold $t$, one can define a corresponding \GTDS $\P$, such that an answer that $\P$ has no solution would provide a decision procedure for the universality question with respect to $\A$ and $t$. We show these connections below.

\begin{restatable}{lemmaStatement}{Automata}
\label{lem:TdsDsa}
For every instance $\P$ of \TDSzo with a discount factor $\lambda$ and a target value $t$, one can compute
in polynomial time a discounted-sum automaton $\A$, such that $\P$ has a solution iff $\A$ is not universal.
\end{restatable}
\begin{proof}
First, we claim that the \TDSzo instance $\P$ is equivalent to the \TDS instance $\P'$ with $a=-t(1-\lambda)$, $\beta=1-t(1-\lambda)$ and $t'=0$. Indeed, subtracting $t(1-\lambda)$ from every element in a discounted-sum sequence $w\in\{0,1\}^\omega$, provides a discounted-sum sequence $w'\in\{-t(1-\lambda),1-t(1-\lambda)\}^\omega$, such that $w'=w- t(1-\lambda)  \sum_{i=0}^\infty \lambda^i = w - \frac{t(1-\lambda)}{1-\lambda} = w - t$. 

Consider the discounted-sum automaton $\A$ of Figure~\ref{fig:DSA} over the alphabet $\{a,b\}$. The automaton $\A$ has exactly two runs over a word $w$ -- a run $r_1$ that is solely in $q_1$ and a run $r_2$ that is solely in $q_2$. The value that $\A$ assigns to $w$, denoted by $\A(w)$, is the minimum between the value of $r_1$ on $w$, denoted by $r_1(w)$ and the value of $r_2$ on $w$, denoted by $r_2(w)$. Note that, by the automaton weights to $a$ and $b$, we have for every word $w$ that $r_1(w) = -r_2(w)$. Hence, for every word $w$, $\A(w)=0$ iff $r_1(w)=0$. 

Now, the infinite universality problem of $\A$ asks whether for all infinite words $w$, we have $\A(w)<0$. Thus, the answer to the universality problem is ``no'' iff there is a sequence $w\in\{a,b\}^\omega$ such that $r_1(w)=0$, which is true iff the answer to the given \TDSzo problem $P$ is ``true''.
\end{proof}

%\begin{proof}[Proof sketch]
%We show that the discounted-sum automaton $\A$ of Figure~\ref{fig:DSA} provides the reduction.
%\end{proof}

\begin{figure}[h!]
%Exported from Xfig in a 50% ratio
\centering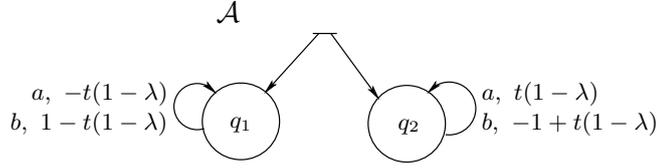 \caption{The discounted-sum automaton, with a discount factor $\lambda$, whose $<$-universality problem is equivalent to a \TDSzo instance with a discount factor $\lambda$ and a target value $t$.}\label{fig:DSA}
\end{figure}

\begin{restatable}{theorem}{TdsDsa}
\label{thm:TdsDsa}
If \TDS is undecidable then so are the universality and inclusion problems of discounted-sum automata over infinite words.
\end{restatable}
\begin{proof}
	Directly follows from Lemma~\ref{lem:TdsDsa}.
\end{proof}

\Skip{Given Theorem~\ref{thm:TdsDsa}, one may wonder whether the opposite direction also holds. For example, whether the positive result regarding \TDS with $\lambda \geq \frac{1}{2}$ can be translated into solving the universality problem of discounted-sum automata with $\lambda \geq \frac{1}{2}$. Unfortunately, this is not the case, due to the additional power of discounted-sum automata, as shown in Example~\ref{exm:RegExp}. 
Still, based on our results on \GTDS we can derive new decidability results for 
restricted some variants of the universality problem for discounted-sum automata.
} % end of \Skip

The next theorem demonstrates the close connection between the target discounted-sum problem and the core difficulty in solving the universality problem. 
\begin{restatable}{theorem}{TdsNegative}
\label{thm:TdsNegative}
Consider an instance $\P$ of \GTDS with a discount factor $\lambda$, weights $\Dm{1} ,\ldots, \Dm{k}$, and a target value $t\in\Rat$. 
If $\P$ has no solution, then the $<$-universality and $\leq$-universality problems with the threshold $t$ are 
decidable over the class of finite-words (resp., infinite-words) \DSA with a discount factor $\lambda$ and weights from $\{\Dm{1} ,\ldots, \Dm{k}\}$.
\end{restatable}

\begin{proof}
Let $\P'$ be the \GTDS problem over natural weights and a target $t'$ that is equivalent to $\P$ (Theorem~\ref{thm:GTdsToGTdszo}).
Consider the tree consisting of all runs of the explorations of $t'$ in base $\frac{1}{\lambda}$, as described in Sections \ref{sec:TdsResults} and \ref{sec:GtdsResults}. 
Since $\P'$ is known to have no solution, each path of the tree is finite. As the degree of every node in the tree is also finite, we get by K\"onig's Lemma that the tree is finite, having some height $H$. 

Now, recall that each run of the exploration stops in one of four cases: i) reaching the target value; ii) repeating a previous gap, which guarantees an eventually-periodic solution; iii) reaching a gap that is too big to be recovered, implying that every continuation of this prefix will be below the threshold; and iv) exceeding the threshold,  implying that every continuation of this prefix will be above the threshold.

As $\P'$ is known to have no solution, all the exploration runs are guaranteed to stop for either reason (iii) or (iv) above. In both cases, all the (finite and infinite) continuations of this finite sequence of weights is known to provide a discounted-sum that is either above or below the threshold.

Hence, the finite and infinite universality problems can be decided by exploring the above finite tree in parallel to running the \DSA $\A$ over input words of length up to $H$. For each such word, we can tell whether all of its continuations are above or below the threshold $t$.
\end{proof}

\section{Connections to Other Areas}\label{sec:Connections}
In this section, we show the connection between \TDS and various areas in mathematics and computer science. In particular, we show that if \TDS is undecidable then so is the reachability problem of one-dimensional piecewise affine maps, and from the other direction, if \TDS is decidable then so is the membership problem of a middle-$k$th Cantor set.

\subsection{Piecewise Affine Maps}\label{sec:PAM}

We shall show that \TDS can be reduced to the reachability problem of one-dimensional piecewise affine maps.

A \emph{Piecewise affine map} (\PAM) of dimension $d$ is a function $f:\Reals^d\to\Reals^d$, where the domain is partitioned into disjoint sections, and from each section there is a different affine map. That is, each affine map is $f(x)=a_i x+b_i$, where $a_i$ and $b_i$ are vectors of dimension $d$, specific to the $i^{th}$ section, and $a_i$ is non-zero. 
%\jotop{We assume that vectors $a_i$ are non-zero at every component. }
% There might be several possible maps for each section, in which case the \PAM is nondeterministic. 

The \emph{reachability problem} asks, given a \PAM $f$, and points $s$ and $t$, whether there exists $n\in\Nat$, such that $f^n(s)=t$. The problem is undecidable for 2, or more, dimensions \cite{KCG94,Moo90,Moo91}. It is open for one dimension, even when only having two sections \cite{AMP95,AG02,KPC08}. (For a single section, meaning when the mapping is affine but not ``piecewise'', the reachability problem is known as the ``orbit problem'', and it is decidable for all dimensions over the rationals \cite{KL86,Sha79}.) 

Next, we establish the connection between \TDS and the above reachability problem.

\begin{restatable}{lemmaStatement}{PAMs}
	\label{lem:TdsPam}
For every instance $\P$ of \TDSzo with a discount factor $\lambda < \frac{1}{2}$ and a target value $t$, one can compute 
a one-dimensional \PAM $f$,  such that $\P$ has a solution iff $1$ is not reachable from $t$ via $f$.
\end{restatable}
\begin{proof}
Let $t = \frac{a}{b}$ and  $\lambda = \frac{p}{q}$.
We may assume that $t \leq \frac{\lambda}{1 - \lambda}$, as otherwise $\P$ has no solution.
We define the following one-dimensional \PAM instance:

\begin{eqnarray*}
f(x) &=& \left\{\begin{array}{ll}
\frac{x}{\lambda} &~ x < \lambda \\[1ex]
\frac{x - \lambda}{\lambda} &~ \lambda \leq x \leq \frac{\lambda}{1-\lambda} \\[1ex]
%2 \lambda &~ x > \frac{\lambda}{1-\lambda} 
{b \cdot q \cdot x} &~ \frac{\lambda}{1-\lambda} < x < 1\\[1ex]
{x} &~ x = 1\\[1ex]
{p \cdot x} &~ 1 < x < 2 \\[1ex]
{x-1} &~ x  \geq 2 
\end{array}\right.
\end{eqnarray*}

Consider the sequence of gaps in exploring the representation of $t$ in base $\frac{1}{\lambda}$, using only \Dz's and \Do's. 
Note that iterations of $f$ exactly match that sequence if it is infinite and extend that sequence of gaps if the exploration is finite. 
\begin{itemize}
\item The initial gap $t$ is the starting point in iterating $f$. 
\item If the current gap is smaller than $0.10^\omega_{[\frac{1}{\lambda}]}$, which equals $\lambda$, it is multiplied by $\frac{1}{\lambda}$ for getting the new gap (as a \Dz is chosen for the current position in the representation).
\item If the current gap is equal to or bigger than $0.10^\omega_{[\frac{1}{\lambda}]}$ and smaller than or equal to $0.1^\omega_{[\frac{1}{\lambda}]}$, which equals $\frac{\lambda}{1-\lambda}$, it is multiplied by $\frac{1}{\lambda}$ and
$\Do$ is subtracted to get the new gap  
$\frac{x - \lambda}{\lambda}$ (as \Do is at the current position in the representation).
\item If the current gap is bigger than $0.1^\omega_{[\frac{1}{\lambda}]}$, which equals $\frac{\lambda}{1-\lambda}$, then the exploration stops, as there is no relevant representation. 
Consider the first $i$ such that    $f^i(t) > \frac{\lambda}{1-\lambda}$. 
Observe that $f^i(t)$ belongs to the interval $(\frac{\lambda}{1-\lambda}, 1)$. Also, $f^i(t)= f^i(\frac{a}{b})$ and it is of the form $\frac{K}{b \cdot p^i}$, where $K \in \Nat$. 
Then, $f^{i+1}(t) = \frac{q \cdot K }{p^{i}} >  \frac{b \cdot q \cdot \lambda}{1-\lambda} = \frac{b \cdot q \cdot p}{q-p}  > 1$. 
\item We can show, by induction on $k$, that for every $t' = \frac{c}{p^k} \geq 1$, some iteration of $f$ on $t'$ reaches $1$. 
Indeed, if $t'$ is a natural number, then
$f^{t'-1}(t') = t' - (t'-1) = 1$. Otherwise, for $n = \floor{t'} - 1$, we have
$f^{n}(t') = t' - n$ belongs to $(1,2)$ and 
$f^{n+1}(t')$ is of the form $\frac{c'}{p^{k-1}} \geq 1$.    
\end{itemize}

Thus, $1$ is reachable from $t$ by iterating $f$ iff at some point in exploring $t$ in base $\frac{1}{\lambda}$ we have to use a digit different from $\Dz$ or $\Do$.
\end{proof}

\begin{restatable}{theorem}{TdsPam}
\label{thm:TdsPam}
If \TDS is undecidable then so is the reachability problem of one-dimensional piecewise affine maps.
\end{restatable}
\begin{proof}
	Directly follows from Lemma~\ref{lem:TdsPam}. 
\end{proof}

\subsection{Cantor Sets}\label{sec:CS}
We shall show that the membership problem of a middle-$k$th Cantor set can be reduced to \TDS.

The \emph{Cantor set} contains the numbers between $0$ and $1$ that are not removed by iteratively removing the middle third: at the first step, the numbers in $(\frac{1}{3},\frac{2}{3})$ are removed; then, the middle third of both the upper and lower parts are removed, and so on.

With base $3$, a number between $0$ and $1$ has a representation with only \Dz's and \Dt's if and only if it is in the Cantor set.
It is easy to check if a rational number is in the Cantor set, since if so, it must have an eventually-periodic representation.

A variation of the Cantor set, where at each step only the $\frac{1}{5}$ upper and lower parts remain, is very analogous. 
In general, for every integer $k>2$, the set of numbers between $0$ and $1$ that are not removed by iteratively removing the middle $k$th is termed the \emph{middle-$k$th Cantor set} \cite{Din01,Eid05,GR95}.

However, removing, for example, the middle $\frac{1}{5}$, i.e., considering the middle-fifth Cantor set, is something very different from removing the middle third -- at each step the remained parts should be multiplied by $\frac{5}{2}$, which makes it analogous to a representation in base $\frac{5}{2}$. 

We show below that a middle-$k$th Cantor set corresponds to the set of numbers that have a representation in base $\frac{2k}{k-1}$ with only \Dz's and \Do's.

\begin{restatable}{lemmaStatement}{Cantor}
	\label{lem:TdsCs}
Consider an integer $k>2$ and a number $t\in [0,1]$. Then $t$ belongs to the middle-$k$th Cantor set iff \TDSzo with a discount factor $\lambda=\frac{k-1}{2k}$ and a target value $t \frac{k-1}{k+1}$  has a solution.
\end{restatable}

\begin{proof}
Consider the set $S$ of numbers that have a $\frac{2k}{k-1}$-representation with only \Dz's and \Do's. By the uniqueness of the representation (Lemma~\ref{lem:UniqueZO}), the set $S$ can be achieved by the limit of the following iterative procedure: 
\begin{enumerate}
\item We start with the set $S_1$ of all numbers between $0.0^\omega$ and $0.1^\omega_{[\frac{2k}{k-1}]}$, which is $[0,\frac{k-1}{k+1}]$.
\item We generate the set $S_2$ by removing from $S_1$ all the numbers that cannot be represented with only \Dz's and \Do's according to the first digit, namely the numbers that are smaller than $0.10^\omega_{[\frac{2k}{k-1}]}$, which equals $\frac{k-1}{2k}$, and bigger than $0.01^\omega_{[\frac{2k}{k-1}]}$, which equals $\frac{k-1}{2k} \cdot \frac{k-1}{k+1}$.
Note that $S_2$ has two separate segments -- $S_2^0 = [0.00^\omega_{[\frac{2k}{k-1}]} , 0.01^\omega_{[\frac{2k}{k-1}]}]$ and $S_2^1 = [0.10^\omega_{[\frac{2k}{k-1}]} , 0.11^\omega_{[\frac{2k}{k-1}]}]$.
\item We generate the set $S_3$ by removing from $S_2$ all the numbers that cannot be represented with only \Dz's and \Do's according to the second digit. That is, we remove from $S_2^0$ the numbers that are smaller than $0.010^\omega_{[\frac{2k}{k-1}]}$ and bigger than $0.001^\omega_{[\frac{2k}{k-1}]}$, and from $S_2^1$ the numbers that are smaller than $0.110^\omega_{[\frac{2k}{k-1}]}$ and bigger than $0.101^\omega_{[\frac{2k}{k-1}]}$.
\item [i)] In the i-th iteration, we generate $S_i$ by removing from $S_{i-1}$ 
all the numbers that cannot be represented with only \Dz's and 
\Do's at the i-th position. In consequence, $S_i$ consists of numbers that have a representation which
up to $i$th position consists of \Dz's and \Do's. 
\end{enumerate}

Then, $S = \bigcap_{i=1}^\infty S_i$.

We claim that every removed segment is exactly the middle $k$th of the segment from which it is removed. We show it for the case of generating $S_2$ from $S_1$, while all other cases are analogous, as they are generated in the exact same way, with just a shift to the right of the representation. Recall that the set $S_1$ is the segment $[0,\frac{k-1}{k+1}]$, and $S_2$ is generated from it by removing the segment $(\frac{k-1}{2k} \cdot \frac{k-1}{k+1},\frac{k-1}{2k})$. 

We first show that the size of the removed segment is $\frac{1}{k}$ of the size of $S_1$. Indeed:
$$
\frac{ \frac{k-1}{2k} -  \frac{k-1}{2k} \cdot \frac{k-1}{k+1}}{ \frac{k-1}{k+1} } = 
%\frac{ \frac{k-1}{2k} \cdot (1 - \frac{k-1}{k+1})} { \frac{k-1}{k+1} } = 
\frac{ \frac{k-1}{2k} \cdot \frac{2}{k+1}} { \frac{k-1}{k+1} } = 
\frac{ 2(k-1)} { 2k(k-1) } = \frac{1}{k}
$$

Next, we show that the removed segment is in the middle, meaning that the size of the lower segment of $S_2$ is $\frac{(k-1)/2}{k}$ of the size of $S_1$. Indeed:
$\frac{ \frac{k-1}{2k} \cdot \frac{k-1}{k+1}}{ \frac{k-1}{k+1} } = \frac{k-1}{2k}$.

Now, we showed that the numbers that have a $\frac{2k}{k-1}$-representation with only \Dz's and \Do's are exactly the numbers in the middle-$k$th Cantor set of the segment $[0, \frac{k-1}{k+1}]$. Hence, by the multiplicative nature of the middle-$k$th removal procedure, a number $t$ is in the middle-$k$th Cantor set of the segment $[0,1]$ iff $t \frac{k-1}{k+1}$ has a $\frac{2k}{k-1}$-representation with only \Dz's and \Do's.
\end{proof}

%\begin{proof}[Proof sketch]
%Consider the set $S$ of numbers that have a $\frac{2k}{k-1}$-representation with only \Dz's and \Do's. By the uniqueness of the representation (Lemma~\ref{lem:UniqueZO}), the set $S$ can be achieved by starting with all numbers between $0.0^\omega$ and $0.1^\omega_{[\frac{2k}{k-1}]}$, and taking the limit of the procedure that iteratively removes the numbers that cannot be represented with only \Dz's and \Do's according to the $n$th digit. 

%We show that along this procedure, every removed segment is exactly the middle $k$th of the segment from which it is removed. 
%\end{proof}

\begin{restatable}{theorem}{TdsCs}
\label{thm:TdsCs}
If \TDS is decidable then so is the membership problem in the middle-$k$th Cantor set.
\end{restatable}
\begin{proof}
	Directly follows from Lemma~\ref{lem:TdsCs}. 
\end{proof}

\begin{remark}
One may wonder why a representation in base $\frac{5}{2}$, for example, that only uses the \Dz and \Do digits is not similar to the standard Cantor set, with the only difference of removing the upper third rather than the middle third. This follows the intuition that, at the $n$th step, we remove the numbers whose $\frac{5}{2}$-representation has a \Dt in the $n$th position. The problem is that it will also remove numbers that do have a representation with only \Dz's and \Do's, as with a nonintegral base, the representation need not be unique.
\end{remark}

\section{Conclusions}
The target discounted-sum problem, which is identified and
analyzed for the first time in this paper, turns out to be
related to several open problems in mathematics and computer science, among
them are problems of $\beta$-expansions, discounted-sum automata
and games, piecewise affine maps, and generalizations of the
Cantor set.  
We established a partial solution to the target discounted-sum problem,
resolving its restrictions to finite and eventually-periodic sequences,
as well as to various specific discount factors, among which are the cases that $\lambda = \frac{1}{n}$,
for every $n \in \Nat$. We generalized our solutions to an extended
version of the target discounted-sum problem, in which there may
be arbitrarily many weights and an $\omega$-regular constraint on
the allowed sequences. Using these generalized solutions, we
solved some open problems on functional discounted-sum automata.

\noindent{\em Acknowledgements.}
%{\footnotesize 
This research was supported in part by the European Research Council (ERC) under grant 267989 (QUAREM) and 
by the Austrian Science Fund (FWF) under grants S11402-N23 (RiSE) and Z211-N23 (Wittgenstein Award). 
%}

\bibliography{bib}

\end{document}